\documentclass[letterpaper,10pt,conference]{ieeeconf}
\IEEEoverridecommandlockouts
\usepackage{amsmath,amssymb}
\usepackage{newpxtext,newpxmath}
\usepackage{microtype}
\usepackage{cite}
\usepackage{booktabs}
\usepackage{graphicx}
\usepackage{xcolor}
\definecolor{AccentGreen}{RGB}{76,175,80}
 
\usepackage[
    colorlinks=true,
    linkcolor=AccentGreen,
    citecolor=AccentGreen,
    urlcolor=AccentGreen,
    filecolor=AccentGreen
]{hyperref}
 
\AtBeginDocument{%
    \hypersetup{
        colorlinks=true,
        linkcolor=AccentGreen,
        citecolor=AccentGreen,
        urlcolor=AccentGreen,
        filecolor=AccentGreen
    }%
}

\makeatletter
\renewcommand{\tagform@}[1]{%
    \maketag@@@{%
        \color{AccentGreen}%
        (\ignorespaces#1\unskip\@@italiccorr)%
    }%
}
 
\renewcommand{\@biblabel}[1]{{\color{AccentGreen}[#1]}}
\makeatother
 
\newtheorem{assumption}{Assumption}
\newtheorem{definition}{Definition}
\newtheorem{lemma}{Lemma}
\newtheorem{theorem}{Theorem}
\newtheorem{corollary}{Corollary}
\newtheorem{remark}{Remark}
 
\DeclareMathOperator{\col}{col}
\DeclareMathOperator{\diag}{diag}
\DeclareMathOperator{\Proj}{Proj}

\newcommand{\R}{\mathbb{R}}

\newcommand{\Dphi}{\Delta_{\phi}}
\newcommand{\DPhi}{\Delta_{\Phi}}
 
\title{\LARGE\bfseries
Adaptive Observer of Nonlinear One-Sided Lipschitz Systems Using
Estimated State Regressors With Finite Excitation}
 
\definecolor{EmailBlue}{RGB}{35,110,180}
\author{Hamid Taghavifar and Brian Delgado Aguilar%
\thanks{H. Taghavifar and B. Delgado Aguilar are with the Department of Mechanical, Industrial and Aerospace Engineering, Concordia University, Montreal, QC H3G 1M8, Canada
(e-mail:
\href{mailto:hamid.taghavifar@concordia.ca}
{\textcolor{EmailBlue}{\textit{hamid.taghavifar@concordia.ca}}};
\href{mailto:brian27ada@gmail.com}
{\textcolor{EmailBlue}{\textit{brian27ada@gmail.com}}}).}%
\thanks{This work has been submitted to the IEEE for possible publication. Copyright may be transferred without notice, after which this version may no longer be accessible.}%
}
\begin{document}

\maketitle
\thispagestyle{empty}
\pagestyle{empty}

\begin{abstract}
For systems with unknown parameters, finite excitation and concurrent learning can potentially yield parameter convergence without persistent excitation but the regressor may still depend on inaccessible states, leading to regressor mismatch. In this paper, this problem is addressed for a class of nonlinear systems with one-sided Lipschitz properties and quadratically inner-bounded nonlinearities with bounded disturbances and linearly parametrized uncertainties. To this aim, an output-integral regression is utilized by using measured outputs and estimated states, and history-stack residual is explicitly bounded in terms of state-estimation error and disturbance. Furthermore, a perturbation bound between the estimated-state and true-state information matrices is derived. Additionally, an OSL-QIB LMI condition is applied for the observer design and a projected adaptive law is designed without needing exact output matching. Stability analysis's results indicate the proposed observer and parameter estimation outperform observers without history-stack learning term.
\end{abstract}

\section{Introduction}

In linearly parametrized dynamical systems, parameter convergence in observers is carried out by applying the persistence of excitation (PE) condition \cite{Annaswamy}. However, PE condition can limit the class of systems due to a need for regressors to stay sufficiently rich during every moving time window \cite{Kreisselmeier1}. When disturbances and weak excitation are present, robust modifications such as leakage and projection can prevent parameter drift, but they do not yield the informative data essential for parameter convergence \cite{CaiDeQueirozDawson2006}.

Finite-excitation methods tackle this limitation by gathering informative data during finite time interval. Concurrent learning (CL) instills stored data in parameter update process \cite{Kamalapurkar2017CL} and integral concurrent learning (ICL) forms integral regressions without state-derivative estimation \cite{Parikh2019}. In adaptive observer design, finite- and initial-excitation methods have been developed for disturbed MIMO LTI systems \cite{Bhasin2,Bhasin1}. Furthermore, nonlinear observers without PE and filtered-regression structures have been studied in
\cite{TomeiMarino2023,Pyrkin2019}.
 
According to the state-of-the-art literature, an output-feedback concurrent-learning observer was developed for second-order nonlinear systems by using estimated state trajectories \cite{Kamalapurkar2017Observer}. Moreover, in \cite{OgriBellKamalapurkar2023}, an incremental-multiplier/LMI observer was added to concurrent learning for affine nonlinear systems. Such related studies indicate that estimated states can be used for concurrent learning. That said, the above studies have not considered the OSL-QIB systems in which the explicit relation between the estimated-state and true-state finite-excitation matrices, disturbance-dependent stored-data residuals, and a combined state-parameter stability condition exists.

The use of Lyapunov and LMI conditions for designing observers for one-sided Lipschitz (OSL), quadratically inner-bounded (QIB), and related non-globally Lipschitz systems was reported by \cite{AbbaszadehMarquez2010,ZemoucheRajamani2022} where mainly state estimation for known nonlinear models was considered without finite-excitation estimation of unknown parameters. On the other hand, the estimated-state concurrent-learning observers discussed above have not considered observer and data-quality conditions based on OSL-QIB incremental inequalities. Combining these two settings introduces an observer-error-dependent residual into the stored output-integral regression. Moreover, the information matrix constructed using estimated states can differ from the unavailable true-state information matrix. Hence, the effects of both differences, together with bounded disturbances, need to be considered in the parameter estimation and stability analyses.
A finite-excitation adaptive observer design for the class of OSL-QIB nonlinear system based on estimated state regressor is proposed in this paper. The major contributions of this paper include: (i) formulating an output integral regression without output differentiation but including the residual due to state estimation error and disturbance in it and obtaining a perturbation inequality for estimated state and true state finite-excitation matrices, (ii) formulating a projected adaptive law by employing the computable estimated state finite-excitation condition along with OSL-QIB LMI with no exact output matching, and (iii) ensuring the regional uniformly ultimately boundedness of state and parameter errors with exponential stability for disturbance-free and residual-free stored data.

\section{Problem Formulation}
\label{sec:problem}

The following nonlinear system is considered:

\begin{equation}
\label{eq:plant}
\begin{aligned}
\dot{x} &= Ax+Bu+\phi(x,u)+\Phi(x,u)\theta+Dd(t) \\
y &= Cx 
\end{aligned}
\end{equation}

where $x\in\R^n$, $u\in\R^m$, and $y\in\R^p$ denote the system state, known input, and measured output, respectively. Moreover, $\theta\in\R^q$ is the unknown constant parameter vector and $d\in\R^{n_d}$ stands for an unknown disturbance. The matrices $A$, $B$, $C$, and $D$ and functions $\phi:\R^n\times\R^m\to\R^n$ and $\Phi:\R^n\times\R^m\to\R^{n\times q}$ are known and the functions are locally Lipschitz with respect to the state.  

\begin{assumption}
\label{ass:compact}

The input $u$ is locally bounded and belongs to the compact set $\mathcal U\subset\R^m$. The disturbance $d$ is also locally bounded and satisfies $\|d(t)\|\leq\bar d$. The solution exists for each $t$ and stays within the known operating region of $\mathcal X\subset\R^n$. Additionally, the unknown parameters are within the known compact convex set below:

\begin{equation}
\label{eq:theta_set}
\Omega_\theta\triangleq\{\vartheta\in\R^q:\|\vartheta\|\le\bar\theta\} 
\end{equation}
For the analysis, we consider $r_e>0$ and a compact set $\hat{\mathcal X}$ such that $x-e\in\hat{\mathcal X}$ for all $x\in\mathcal X$ and $\|e\|\leq r_e$. The following incremental inequalities are needed on:
\begin{equation}
\label{eq:design_domain}
\begin{aligned}
\mathcal D_e\triangleq\{(x,\hat x,u):{}&x\in\mathcal X,\quad
\hat x\in\hat{\mathcal X}u\in\mathcal U,\quad \|x-\hat x\|\le r_e\} 
\end{aligned}
\end{equation}
If the inequalities hold globally, $r_e=\infty$ and $\hat{\mathcal X}=\R^n$ can be considered. For the estimation error $e\triangleq x-\hat x$, we define
\begin{equation}
\label{eq:increments}
\Dphi\triangleq\phi(x,u)-\phi(\hat x,u),\qquad
\DPhi\triangleq\Phi(x,u)-\Phi(\hat x,u)
\end{equation}
\end{assumption}

\begin{assumption}
\label{ass:oslqib}
The following conditions are satisfied by the nonlinear increments below for or all $(x,\hat x,u)\in\mathcal D_e$:
\begin{align}
e^\top\Dphi &\le \rho\|e\|^2 \label{eq:osl}\\
\|\Dphi\|^2 &\le \alpha\|e\|^2+\beta e^\top\Dphi \label{eq:qib}\\
\|\Dphi\| &\le \ell_\phi\|e\| \label{eq:local_phi}
\end{align}
where $\rho,\alpha,\beta\in\R$ and $\ell_\phi$ are known constants. Furthermore, the first two conditions are considered in observer design and last inequality condition is merely utilized for obtaining the bound for the finite-window data residual.
\end{assumption}

\begin{assumption}
\label{ass:Phi_lip}
There are known positive constants $\bar\Phi$ and $\ell_\Phi$ such that for all $(x,\hat x,u)\in\mathcal D_e$ we have:
\begin{align}
\|\Phi(\hat x,u)\|&\le\bar\Phi  \label{eq:Phi_bound}\\
\|\Phi(x,u)-\Phi(\hat x,u)\|&\le\ell_\Phi\|x-\hat x\|  \label{eq:Phi_lip}
\end{align}
\end{assumption}
The following condition yields convex observer design condition. Let $v\in\R^n$ standing for $\Dphi$ and $w\in\R^n$ denoting $\DPhi\theta$. First, we define $\kappa_\Phi\triangleq\ell_\Phi\bar\theta$. Subsequently, lets consider augmented vector $\xi\triangleq\col(e,v,w,d)$ as well as symmetric matrices below:
\begin{align}
\mathcal M_1&\triangleq
\begin{bmatrix}
2\rho I&-I&0&0\\
-I&0&0&0\\
0&0&0&0\\
0&0&0&0
\end{bmatrix} \label{eq:M1}\\
\mathcal M_2&\triangleq
\begin{bmatrix}
\alpha I&\frac{\beta}{2}I&0&0\\
\frac{\beta}{2}I&-I&0&0\\
0&0&0&0\\
0&0&0&0
\end{bmatrix} \label{eq:M2}\\
\mathcal M_3&\triangleq
\begin{bmatrix}
\kappa_\Phi^2 I&0&0&0\\
0&0&0&0\\
0&0&-I&0\\
0&0&0&0
\end{bmatrix} \label{eq:M3}
\end{align}
where zero and identity blocks have $\xi$ induced dimensions.
 \begin{assumption}
\label{ass:observer_lmi}
There are $P=P^\top$, $Y\in\R^{n\times p}$, $a_e$, $a_d$, and $\tau_j$ for $j\in\{1,2,3\}$ such that we can have the following:
\begin{equation}
\label{eq:observer_lmi}
\mathcal Q(P,Y,a_e,a_d)+\sum_{j=1}^3\tau_j\mathcal M_j\preceq0
\end{equation}
where
\begin{equation}
\label{eq:Xi_observer}
\Xi\triangleq PA+A^\top P-YC-C^\top Y^\top+a_eI 
\end{equation}
and
\begin{equation}
\label{eq:Q_observer}
\mathcal Q\triangleq
\begin{bmatrix}
\Xi&P&P&PD\\
P&0&0&0\\
P&0&0&0\\
D^\top P&0&0&-a_dI
\end{bmatrix}
\end{equation}
\end{assumption}
Then, the observer gain is obtained as $L=P^{-1}Y$. In comparison with the Young's inequality formulation which involves $P^2$, \eqref{eq:observer_lmi} is affine in all decision variables and conceivably is an actual LMI.

\begin{lemma}
\label{lem:observer_dissipation}
With assumptions~\ref{ass:compact}-\ref{ass:observer_lmi}, we can define
$\eta_0\triangleq(A-LC)e+\Dphi+\DPhi\theta+Dd$. Then the observer gain $L=P^{-1}Y$ satisfies in $\mathcal D_e$ the following:
\begin{equation}
\label{eq:observer_dissipativity}
2e^\top P\eta_0\le-a_e\|e\|^2+a_d\|d\|^2
\end{equation}
\end{lemma}
\begin{proof}
From the OSL and QIB conditions, we have positive definite
$\xi^\top\mathcal M_1\xi$ and
$\xi^\top\mathcal M_2\xi$. Also, $
\|\DPhi\theta\|\le\ell_\Phi\bar\theta\|e\|
=\kappa_\Phi\|e\|,$ which yields $\xi^\top\mathcal M_3\xi\ge0$. By multiplying \eqref{eq:observer_lmi} from both sides by $\xi^\top$ and $\xi$, respectively, it gives
\[
\xi^\top\mathcal Q\xi
\le-\sum_{j=1}^3\tau_j\xi^\top\mathcal M_j\xi\le0.
\]
By using $Y=PL$, $v=\Dphi$, and $w=\DPhi\theta$, the left hand side can be written as follows, \eqref{eq:observer_dissipativity} is obtained to yield:
\begin{align*}
\xi^\top\mathcal Q\xi
={}&2e^\top P\big[(A-LC)e+\Dphi+\DPhi\theta+Dd\big]\\
&+a_e\|e\|^2-a_d\|d\|^2
\end{align*}
\end{proof}
 
\section{Observer and Estimated-State Learning Identity}
\label{sec:identity}

The observer is considered as follows:
\begin{equation}
\label{eq:observer}
\dot{\hat x}=A\hat x+Bu+\phi(\hat x,u)+\Phi(\hat x,u)\hat\theta
+L(y-C\hat x) 
\end{equation}
The observer's error dynamics is derived as following by defining $\tilde\theta\triangleq\theta-\hat\theta$:
\begin{equation}
\label{eq:error_dynamics}
\dot e=(A-LC)e+\Dphi+\Phi(\hat x,u)\tilde\theta
+\DPhi\theta+Dd 
\end{equation}
The known bounded design map $\Psi:\hat{\mathcal X}\times\mathcal U\to\R^{p\times q}$ can be selected to derive the implementable parameter update. Accordingly, the following terms are defined:
\begin{align}
\mathcal M_\Psi(\hat x,u)&\triangleq
P\Phi(\hat x,u)-C^\top\Psi(\hat x,u) \label{eq:Mpsi}\\
\bar m_\Psi&\triangleq
\sup_{(\hat x,u)\in\hat{\mathcal X}\times\mathcal U}
\|\mathcal M_\Psi(\hat x,u)\| \label{eq:mpsi_bound}
\end{align}
\begin{remark}
Let $\Pi_C\triangleq C^\top(C^\top)^\dagger$. Since $I-\Pi_C$ is an
orthogonal projector, for every $\Psi$,
$
\|(I-\Pi_C)(P\Phi-C^\top\Psi)\|
\leq\|P\Phi-C^\top\Psi\|$ which means $\|(I-\Pi_C)P\Phi\| \leq \|P\Phi-C^\top\Psi\|
$. The choice $\Psi^\star=(C^\top)^\dagger P\Phi$ brings this bound
pointwise and thus minimizes $\bar m_\Psi$. The exact matching is
possible if and only if $(I-\Pi_C)P\Phi=0$ in the designing domain. Thus, $\Psi^\star$ gives the least restrictive condition with respect to $\Psi$
\eqref{eq:coupled_gain_condition}; under exact matching, this condition
holds for every $k_c$.
\end{remark}

For the sake of brevity, the known nominal vector field is defined by $\label{eq:f0}
f_0(\zeta,u)\triangleq A\zeta+Bu+\phi(\zeta,u)$.

\begin{definition}
\label{def:data_pair}
The following data pair is defined for a window length $\Delta$ and a sampling instant $t_i\ge\Delta$: 
\begin{align}
\mathcal Y_i&\triangleq y(t_i)-y(t_i-\Delta)
-\int_{t_i-\Delta}^{t_i}C f_0(\hat x(\tau),u(\tau))d\tau
\label{eq:Y_i}\\
\mathcal G_i&\triangleq
\int_{t_i-\Delta}^{t_i}C\Phi(\hat x(\tau),u(\tau))d\tau \label{eq:G_i}
\end{align}
Therefore, $\mathcal Y_i\in\R^p$ and $\mathcal G_i\in\R^{p\times q}$ can be derived by utilizing $u$, $y$, and $\hat x$.
\end{definition}

\begin{lemma}
\label{lem:regression}
When the corresponding data window lies in the design domain, we can have:
\begin{equation}
\label{eq:regression}
\mathcal Y_i=\mathcal G_i\theta+r_i
\end{equation}
where
\begin{equation}
\label{eq:r_i}
r_i=\int_{t_i-\Delta}^{t_i}C\big[Ae(\tau)+\Dphi(\tau)
+\DPhi(\tau)\theta+Dd(\tau)\big]d\tau
\end{equation}
Furthermore
\begin{equation}
\label{eq:r_i_bound}
\|r_i\|\le
c_e\int_{t_i-\Delta}^{t_i}\|e(\tau)\|d\tau
+c_d\int_{t_i-\Delta}^{t_i}\|d(\tau)\|d\tau
\end{equation}
where $ c_e\triangleq\|C\|\big(\|A\|+\ell_\phi+\ell_\Phi\bar\theta\big),
\,\,\,\, c_d\triangleq\|CD\|$.
 
\end{lemma}

\begin{proof}
Due to absolute continuity of $x$, one can write:
\begin{align*}
y(t_i)-y(t_i-\Delta)
&=\int_{t_i-\Delta}^{t_i}C\dot x(\tau)d\tau
=\int_{t_i-\Delta}^{t_i}C f_0(x,u)d\tau\\
&\quad+\int_{t_i-\Delta}^{t_i}C\big[\Phi(x,u)\theta+Dd\big]d\tau
\end{align*}
By subtracting the integral term in \eqref{eq:Y_i}, by using $x=\hat x+e$, and also adding and subtracting $C\Phi(\hat x,u)\theta$, equations \eqref{eq:regression}-\eqref{eq:r_i} are derived. Subsequently, by leveraging the triangle inequality together with Assumptions~\ref{ass:oslqib}-\ref{ass:Phi_lip} gives us:
\begin{align*}
\|r_i\|
&\le c_e\int_{t_i-\Delta}^{t_i}\|e(\tau)\|d\tau
+c_d\int_{t_i-\Delta}^{t_i}\|d(\tau)\|d\tau
\end{align*}
which is \eqref{eq:r_i_bound} and completes the proof.
\end{proof}

\section{Finite-Excitation Adaptive Law}
\label{sec:adaptive_law}

Consider the history stack $\mathcal I_N=\{t_1,\ldots,t_N\}$ and the following matrix being defined as follows:
\begin{equation}
\label{eq:S_hat}
\hat S_N\triangleq\sum_{i=1}^N\mathcal G_i^\top\mathcal G_i
\end{equation}

\begin{definition}
\label{def:es_fe}
The integral-regression stack is called finitely exciting if for some $N\in\mathbb N$ and a known constant $\sigma_N>0$, it satisfies following condition:
\begin{equation}
\label{eq:fe_condition}
\sigma_N I_q \preceq \hat{S}_N
\end{equation}
The above condition implies the finite-data informativity of the stored integral regressions which can be assessed by using the recorded $\mathcal G_i$ directly.
\end{definition}

\begin{lemma}
\label{lem:estimated_true_fe}
Defining
\begin{equation}
\label{eq:true_stack}
\begin{aligned}
\mathcal G_i^x&\triangleq
\int_{t_i-\Delta}^{t_i}C\Phi(x(\tau),u(\tau))d\tau, \,
S_N^x\triangleq\sum_{i=1}^N(\mathcal G_i^x)^\top\mathcal G_i^x
\end{aligned}
\end{equation}
Let
\begin{align}
\delta_i\triangleq\|C\|\ell_\Phi
\int_{t_i-\Delta}^{t_i}\|e(\tau)\|d\tau ;\,\, 
\rho_G\triangleq\sum_{i=1}^N
\big(2\|\mathcal G_i\|\delta_i+\delta_i^2\big)  
\end{align}
Then
\begin{equation}
\label{eq:stack_perturbation}
\|S_N^x-\hat S_N\|\le\rho_G,\qquad
\lambda_{\min}(\hat S_N)-\rho_G\le \lambda_{\min}(S_N^x)
\end{equation}
Hence, for any preferable positive $\sigma_x$:
\begin{equation}
\label{eq:true_fe_margin}
\rho_G+\sigma_x \le \lambda_{\min}(\hat S_N)
\quad\Longrightarrow\quad
\sigma_xI_q\preceq S_N^x 
\end{equation}
\end{lemma}
\vspace{0.2 cm}
\begin{proof}
Let $E_i\triangleq\mathcal G_i^x-\mathcal G_i$. From \eqref{eq:Phi_lip}, $
\|E_i\|\le\int_{t_i-\Delta}^{t_i}
\|C\|\|\DPhi(\tau)\|d\tau\le\delta_i$. As $\mathcal G_i^x=\mathcal G_i+E_i$, the difference between the two stack matrices can be written as $\mathcal G_i^x=\mathcal G_i+E_i$, then$
S_N^x-\hat S_N
=\sum_{i=1}^N
\big(\mathcal G_i^\top E_i+E_i^\top\mathcal G_i+E_i^\top E_i\big).
$ 
We can obtain the following by utilizing the submultiplicativity and the triangle inequality
$\|S_N^x-\hat S_N\|
\le\sum_{i=1}^N\big(2\|\mathcal G_i\|\|E_i\|+\|E_i\|^2\big)
\le\rho_G.
$
Since the second inequality in \eqref{eq:stack_perturbation} is derived from Weyl's eigenvalue perturbation inequality, and consequently, \eqref{eq:true_fe_margin} holds.
\end{proof}

\begin{remark}
The stored pairs $(\mathcal Y_i,\mathcal G_i)$ and estimated-state finite-excitation condition \eqref{eq:fe_condition} are computable online. In contrast, $S_N^x$, $\rho_G$, $R_N$, and real state error terms all leverage of unavailable true state data and are utilized only for analysis and offline verifications. 
\end{remark}

Points will be added to candidates when doing so will add to the rank of the information matrix because its minimum eigenvalue is zero until it achieves full column rank. Thereafter, points will be added or replaced by existing points only when the minimum eigenvalue increases. The stack is frozen when \eqref{eq:fe_condition} is first
satisfied. Denoting its number of stored points by $N$, the activation
time is
\begin{equation}
\label{eq:TF}
T_F\triangleq\max_{1\leq i\leq N}t_i
\end{equation}
Thus, the stored $\mathcal Y_i$, $\mathcal G_i$, and aggregate residual
remain constant for $t\geq T_F$. Since
$\tilde y\triangleq y-C\hat x=Ce$ is available from measurement, the error for each stored regression is defined by $ \varepsilon_i(\hat\theta)
\triangleq\mathcal Y_i-\mathcal G_i\hat\theta
$. The projected parameter update used in this work is then
\begin{align}
\nu(t)&\triangleq\Gamma\Psi(\hat x,u)^\top\tilde y
+\chi_F(t)k_c\Gamma\sum_{i=1}^N
\mathcal G_i^\top\varepsilon_i(\hat\theta)
\label{eq:update_direction}\\
\dot{\hat\theta}&=\Proj_{\Omega_\theta}(\hat\theta,\nu(t))
\label{eq:adapt_law}
\end{align}
In the above equations, $\Gamma=\Gamma^\top\succ0$ and $k_c>0$. Also, $\chi_F(t)=0$ before $T_F$, while $\chi_F(t)=1$ for $t\ge T_F$. We use the standard projection operator in \cite{CaiDeQueirozDawson2006} and start the update with $\hat\theta(0)\in\Omega_\theta$. Therefore, $\hat\theta(t)$ remains inside $\Omega_\theta$, and for every $\theta\in\Omega_\theta$, the following property holds:
 
\begin{equation}
\label{eq:proj_prop}
0 \le
\tilde{\theta}^{\top}\Gamma^{-1}
\big[\Proj_{\Omega_\theta}(\hat{\theta},\nu)-\nu\big]
\end{equation}

\begin{lemma}
\label{lem:param_ineq}
After activation of the history stack, i.e., for every $t\ge T_F$, the update in \eqref{eq:adapt_law} leads to
\begin{equation}
\label{eq:param_ineq}
\begin{aligned}
2\tilde\theta^\top\Gamma^{-1}\dot{\tilde\theta}
\le{}&-2\tilde\theta^\top\Psi(\hat x,u)^\top Ce
-k_c\sigma_N\|\tilde\theta\|^2
+\frac{k_c}{\sigma_N}\|R_N\|^2
\end{aligned}
\end{equation}
where the effects of all stored residuals are collected in
\begin{equation}
\label{eq:R_N}
R_N\triangleq\sum_{i=1}^N\mathcal G_i^\top r_i
\end{equation}
\end{lemma}

\begin{proof}
From \eqref{eq:regression}, the error for each stored point can be written as $
\mathcal Y_i-\mathcal G_i\hat\theta
=\mathcal G_i\tilde\theta+r_i$
Moreover, $\dot{\tilde\theta}=-\dot{\hat\theta}$. Therefore, using the projection property together with
$\hat S_N\succeq\sigma_NI_q$ gives, for $t\ge T_F$,
\begin{align*}
2\tilde\theta^\top\Gamma^{-1}\dot{\tilde\theta}
&\le-2\tilde\theta^\top\Psi^\top Ce
-2k_c\tilde\theta^\top\hat S_N\tilde\theta
-2k_c\tilde\theta^\top R_N\\
&\le-2\tilde\theta^\top\Psi^\top Ce
-2k_c\sigma_N\|\tilde\theta\|^2
+2k_c\|\tilde\theta\|\|R_N\|
\end{align*}
For the last term, Young's inequality can be applied as $
2k_c\|\tilde\theta\|\|R_N\|
\le k_c\sigma_N\|\tilde\theta\|^2
+\frac{k_c}{\sigma_N}\|R_N\|^2
$. By replacing the last term by this bound gives
\eqref{eq:param_ineq}.
\end{proof}

\begin{lemma}
\label{lem:RN_bound}
To cover all time intervals,lets define $
\mathcal W_N\triangleq
\bigcup_{i=1}^N[t_i-\Delta,t_i]$. Then, the aggregate residual satisfies:
\begin{equation}
\label{eq:RN_bound}
\|R_N\|^2\le
\chi_e\sup_{\tau\in\mathcal W_N}\|e(\tau)\|^2
+\chi_d\sup_{\tau\in\mathcal W_N}\|d(\tau)\|^2
\end{equation}
where one explicit choice of the constants is
\begin{equation}
\label{eq:chied}
\begin{aligned}
\bar g\triangleq\Delta\|C\|\bar\Phi, \,
\chi_e\triangleq2N^2\bar g^2c_e^2\Delta^2, \,
\chi_d\triangleq2N^2\bar g^2c_d^2\Delta^2
\end{aligned}
\end{equation}
\end{lemma}

\begin{proof}
Let
$E_N\triangleq\sup_{\tau\in\mathcal W_N}\|e(\tau)\|$ and
$D_N\triangleq\sup_{\tau\in\mathcal W_N}\|d(\tau)\|$.
Equations~\eqref{eq:Phi_bound} and \eqref{eq:r_i_bound} give
$\|\mathcal G_i\|\leq\bar g$ and
$\|r_i\|\leq\Delta(c_eE_N+c_dD_N)$. Therefore,
\[
\|R_N\|\leq\sum_{i=1}^N\|\mathcal G_i\|\|r_i\|
\leq N\bar g\Delta(c_eE_N+c_dD_N)
\]
By squaring this inequality and applying $(a+b)^2\leq2a^2+2b^2$ yields
\eqref{eq:RN_bound}.
\end{proof}

\begin{figure*}[!t]
    \centering \includegraphics[width=0.31\textwidth] {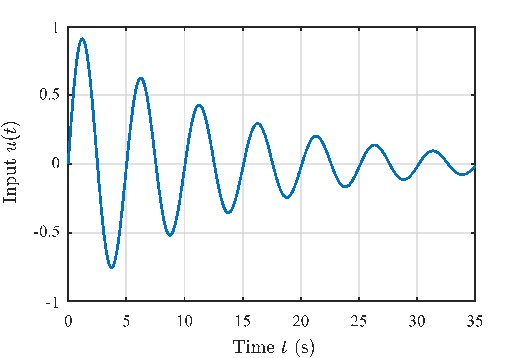}
    \hfill   \includegraphics[width=0.31\textwidth]{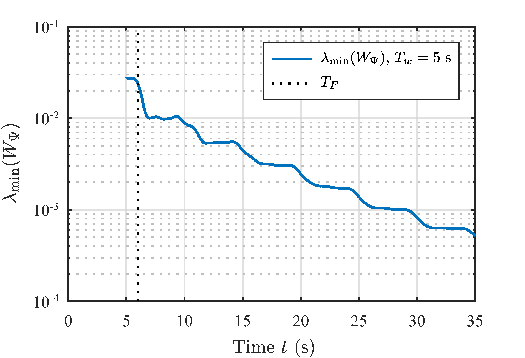}
    \hfill \includegraphics[width=0.31\textwidth]   {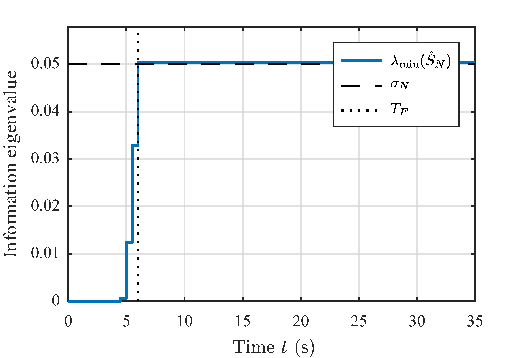}
    \caption{ Left: decaying input. Center: minimum eigenvalue of the $5$s sliding window Gramian. Right: estimated-state finite-excitation certificate.}
    \label{fig:excitation_summary}
\end{figure*}

\begin{figure*}[!t]
\centering
\includegraphics[width=0.32\textwidth]
{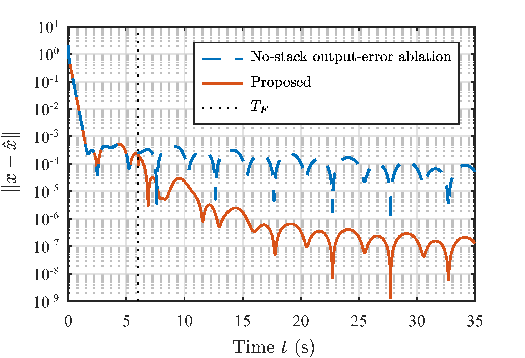}
\hfill
\includegraphics[width=0.32\textwidth]
{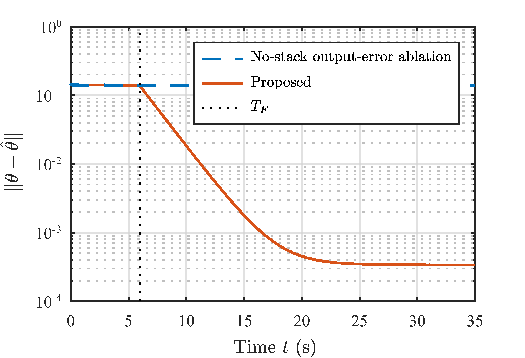}
\hfill
\begin{minipage}[c]{0.30\textwidth}
\centering
\scriptsize
\setlength{\tabcolsep}{2pt}
\renewcommand{\arraystretch}{1.12}
\begin{tabular}{@{}lcc@{}}
\toprule
Method & $J_{e,F}$ & $J_\theta$\\
\midrule
No stack & $1.7525{\times}10^{-4}$ & $1.3891{\times}10^{-1}$\\
Proposed & $2.1322{\times}10^{-5}$ & $3.3858{\times}10^{-4}$\\
\midrule
Factor & $8.22$ & $410.28$\\
\bottomrule
\end{tabular}
\end{minipage}
\caption{Estimation comparison: state error (left), parameter error
(center), and post-activation/terminal metrics (right).}
\label{fig:estimation_errors}
\end{figure*}

\begin{figure*}[!t]
    \centering
    \includegraphics[width=0.31\textwidth]
    {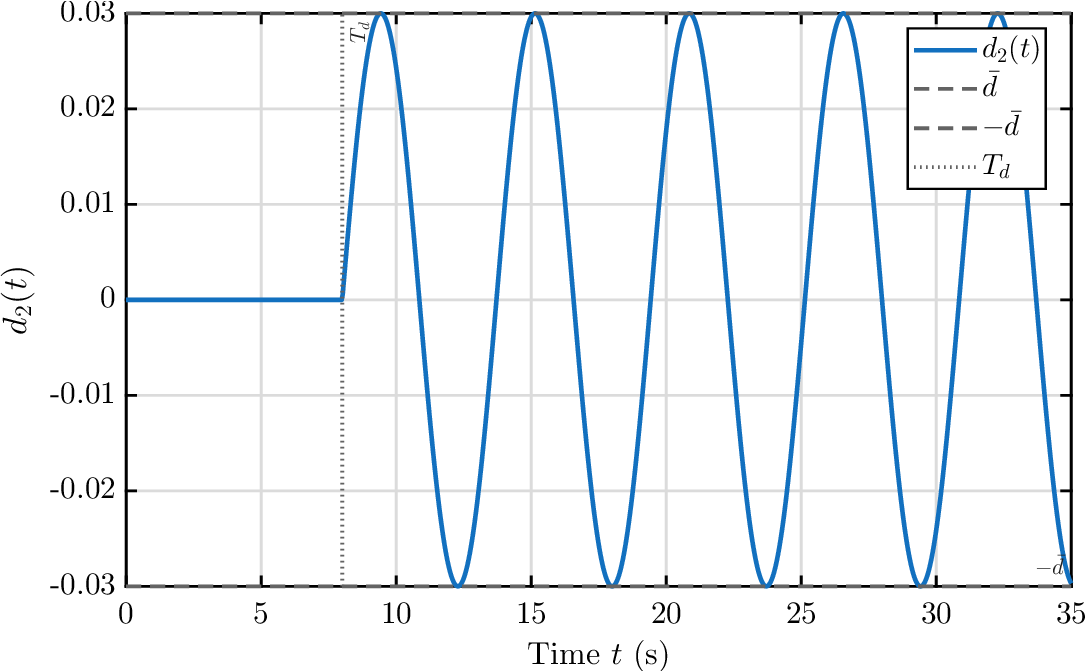}
    \hfill
    \includegraphics[width=0.31\textwidth]
    {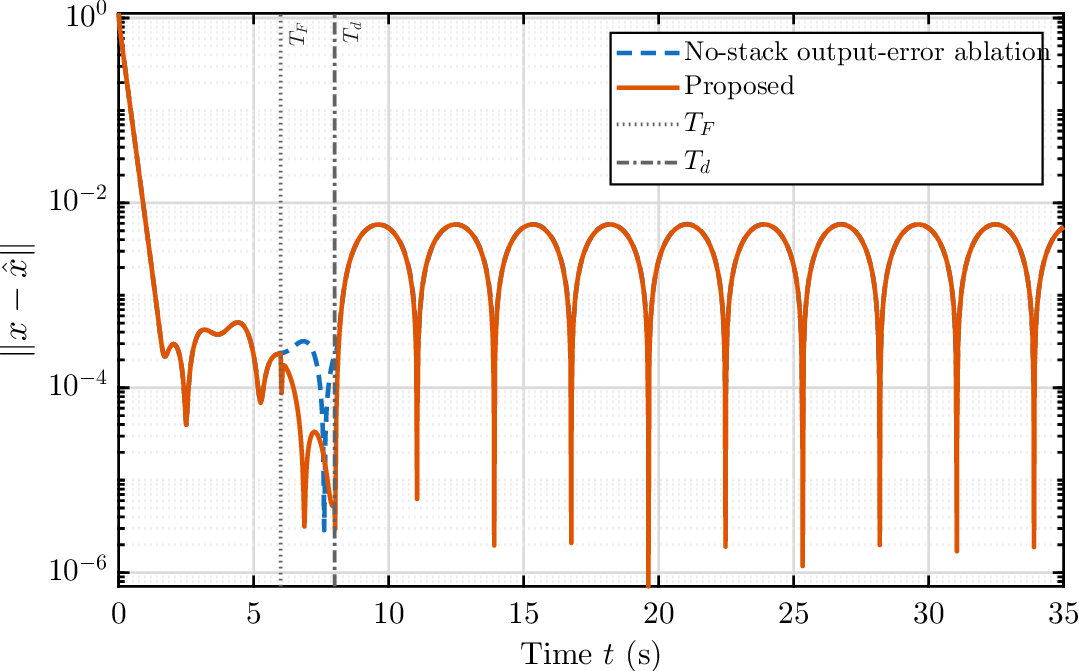}
    \hfill
    \includegraphics[width=0.31\textwidth]
    {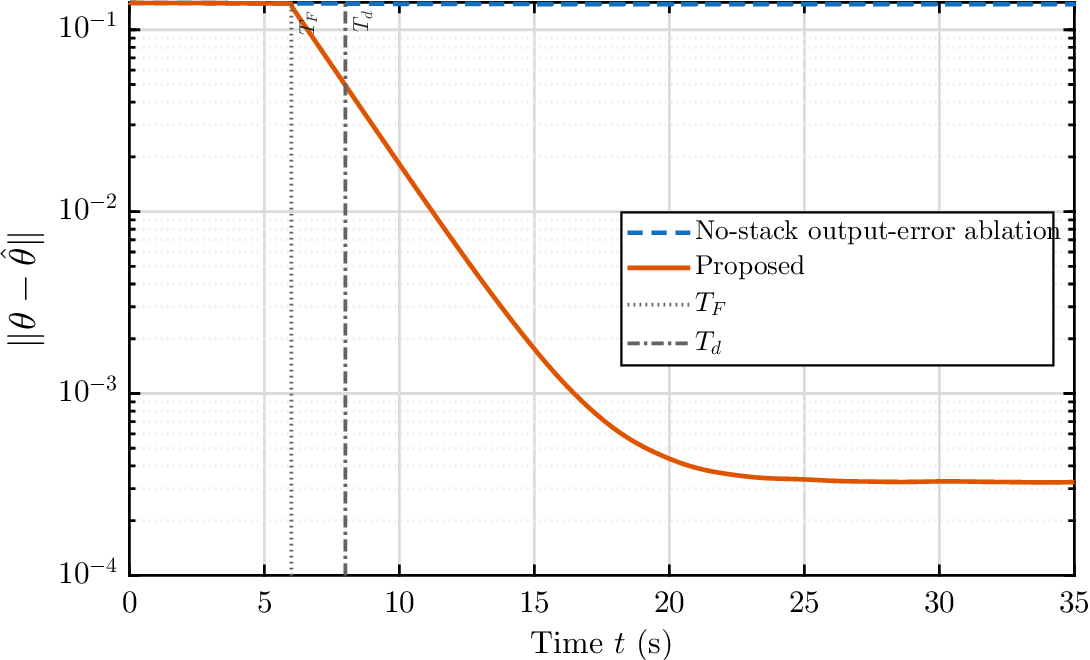}
    \caption{Disturbed-case validation with $\bar d=0.03$ and
    $T_d=8$s. Left: applied bounded disturbance. Middle:
    state-estimation errors. Right: parameter-estimation errors.
    The history stack is frozen at $T_F=6s$ before the disturbance
    is introduced.}
    \label{fig:disturbed_validation}
\end{figure*}

\section{Stability Analysis}
\label{sec:stability}

Here, the stability analyses is performed by considering the following Lyapunov candidate function:

\begin{equation}
\label{eq:V}
V=e^\top Pe+\tilde\theta^\top\Gamma^{-1}\tilde\theta
\end{equation}

For $z\triangleq\col(e,\tilde\theta)$, the bounds hold as follows:
\begin{equation}
\begin{aligned}
\label{eq:V_bounds}
\underline{\lambda}\|z\|^2
&\le V \le \bar{\lambda}\|z\|^2\\
\underline\lambda\triangleq
&\min\{\lambda_{\min}(P),\lambda_{\min}(\Gamma^{-1})\}\\
\bar\lambda\triangleq
&\max\{\lambda_{\max}(P),\lambda_{\max}(\Gamma^{-1})\}
\end{aligned}
\end{equation}

The imperfect cancellation of the state-error cross term by the measurable output-error update is represented by:

\begin{equation}
\label{eq:QN}
Q_N\triangleq
\begin{bmatrix}
a_e&-\bar m_\Psi\\
-\bar m_\Psi&k_c\sigma_N
\end{bmatrix}
\end{equation}

The condition:
\begin{equation}
\label{eq:coupled_gain_condition}
0\prec Q_N
\quad\Longleftrightarrow\quad
\bar{m}_{\Psi}^{\,2}<a_e k_c\sigma_N
\end{equation}
 
can be dictated by selecting $k_c$ after the stack is verified.  
\begin{theorem}
\label{thm:main}
Suppose assumptions~\ref{ass:compact}-\ref{ass:observer_lmi} are held. Then, fixed stack satisfies \eqref{eq:fe_condition}, also condition \eqref{eq:coupled_gain_condition} is held. Consider known constants $\bar e_s,\bar d_s\ge0$ so:

\begin{equation}
\label{eq:stack_error_bound}
\begin{aligned}
\sup_{\tau\in\mathcal W_N}\|e(\tau)\|&\le\bar e_s, \,
\sup_{\tau\in\mathcal W_N}\|d(\tau)\|\le\bar d_s
\end{aligned}
\end{equation}

On the stated data domain, $\bar e_s=r_e$ and $\bar d_s=\bar d$ can be always used as conservative choice. Also, by defining

\begin{align}
\lambda_Q&\triangleq\lambda_{\min}(Q_N); \,
\mu\triangleq\frac{\lambda_Q}{\bar\lambda}; \,
c_s\triangleq a_d\bar d^2
+\frac{k_c}{\sigma_N}
(\chi_e\bar e_s^2+\chi_d\bar d_s^2) \label{eq:cs}
\end{align}

For the regional design domain, the below condition is:

\begin{equation}
\label{eq:invariance_condition}
\max\left\{V(T_F),\frac{c_s}{\mu}\right\}
<\lambda_{\min}(P)r_e^2
\end{equation}

Then the observer error stays in $\mathcal D_e$ for all $t\ge T_F$ and:

\begin{equation}
\label{eq:V_comparison}
\begin{aligned}
V(t)\le{}&e^{-\mu(t-T_F)}V(T_F)+\frac{c_s}{\mu}\big(1-e^{-\mu(t-T_F)}\big),
\quad t\ge T_F
\end{aligned}
\end{equation}

Hence,
\begin{equation}
\label{eq:z_transient_bound}
\begin{aligned}
\|z(t)\|^2
\le{}&\frac{\bar\lambda}{\underline\lambda}
e^{-\mu(t-T_F)}\|z(T_F)\|^2+\frac{c_s}{\underline\lambda\mu}
\big(1-e^{-\mu(t-T_F)}\big)
\end{aligned}
\end{equation}

\begin{equation}
\label{eq:ultimate_radius}
\limsup_{t\to\infty}\|z(t)\|
\le\sqrt{\frac{c_s}{\underline\lambda\mu}}
\end{equation}

Thus, errors are ultimately uniformly bounded. If Assumptions~\ref{ass:oslqib}-\ref{ass:Phi_lip}, the LMI implication, and \eqref{eq:mpsi_bound} are globally, the condition \eqref{eq:invariance_condition} is no more needed and the result gets global.

\end{theorem}

\begin{proof}
Along the projected dynamics, $V$ is locally absolutely continuous and its derivative exists almost everywhere. By defining $
\eta_e\triangleq(A-LC)e+\Dphi+\DPhi\theta+Dd,
$ the error dynamics gives us $
\dot V
=2e^\top P\eta_e
+2e^\top P\Phi(\hat x,u)\tilde\theta
+2\tilde\theta^\top\Gamma^{-1}\dot{\tilde\theta}
$. By invoking Lemmas~\ref{lem:observer_dissipation} and \ref{lem:param_ineq} results in
\begin{align*}
\dot V
&\le-a_e\|e\|^2+a_d\|d\|^2
+2e^\top\mathcal M_\Psi(\hat x,u)\tilde\theta\\
&\quad-k_c\sigma_N\|\tilde\theta\|^2
+\frac{k_c}{\sigma_N}\|R_N\|^2\le-a_e\|e\|^2-k_c\sigma_N\|\tilde\theta\|^2
\\
&+2\bar m_\Psi\|e\|\|\tilde\theta\|+a_d\|d\|^2+\frac{k_c}{\sigma_N}\|R_N\|^2
\end{align*}

Based on \eqref{eq:QN}, let $\zeta_z\triangleq\col(\|e\|,\|\tilde\theta\|)$. Then we have:

\begin{align*}
a_e\|e\|^2+k_c\sigma_N\|\tilde\theta\|^2
-2\bar m_\Psi\|e\|\|\tilde\theta\|
&=\zeta_z^\top Q_N\zeta_z\ge\lambda_Q\|z\|^2
\end{align*}

Lemma~\ref{lem:RN_bound}, \eqref{eq:stack_error_bound}, and $\|d(t)\|\le\bar d$ resulting in

\begin{equation}
\label{eq:Vdot_main}
\dot V\le-\lambda_Q\|z\|^2+c_s\le-\mu V+c_s
\end{equation}

for almost every $t\ge T_F$ as long as the solution stays inside the designing domain. By applying the comparison lemma, \eqref{eq:V_comparison} is obtained. For the regional case, \eqref{eq:V_comparison} together with \eqref{eq:invariance_condition} yields $V(t)<\lambda_{\min}(P)r_e^2$. Because $e^\top Pe\le V$, the error cannot get to $\|e\|=r_e$. Therefore, a standard continuation argument implies forward invariance for $t\ge T_F$. Finally, by combining \eqref{eq:V_comparison} and \eqref{eq:V_bounds} it yields \eqref{eq:z_transient_bound} and \eqref{eq:ultimate_radius}.
\end{proof}

\begin{corollary}
\label{cor:exp}
On the conditions of Theorem~\ref{thm:main}, consider $d(t)\equiv0$. Then $\bar d=\bar d_s=0$ and
\begin{equation}
\label{eq:dist_free_bound}
\limsup_{t\to\infty}\|z(t)\|
\le
\sqrt{\frac{k_c\chi_e}{\underline\lambda\mu\sigma_N}}\,\bar e_s
\end{equation}

If fixed stack has zero aggregate residual $R_N=0$:

\begin{equation}
\label{eq:exp_conv}
\|z(t)\|
\le\sqrt{\frac{\bar\lambda}{\underline\lambda}}
\|z(T_F)\|e^{-\mu(t-T_F)/2},\qquad t\ge T_F
\end{equation}
\end{corollary}

\begin{proof}
The first result stems from \eqref{eq:ultimate_radius} by putting $d=0$. For the second one, consider the derivative inequality before applying Lemma~\ref{lem:RN_bound}. In absence of disturbance and aggregate residual, proof of Theorem~\ref{thm:main} yields $
\dot V\le-\lambda_Q\|z\|^2\le-\mu V
$, thus $V(t)\le V(T_F)e^{-\mu(t-T_F)}$, and \eqref{eq:exp_conv} follows from \eqref{eq:V_bounds}.
\end{proof}

\begin{remark}

The excitation test \eqref{eq:fe_condition} is directly achievable even when $\bar e_s$ is unknown. But, a small output residual $\|y-C\hat x\|=\|Ce\|$ does not limit $\|e\|$ if $C$ has nullspace of nontrivial. Thus, an output-residual threshold confirms \eqref{eq:stack_error_bound} by utilizing a state-error certificate alone.
\end{remark}

 \section{Numerical Results and Discussion}
\label{sec:simulation_results}

In this section, the proposed observer's results are synthesized numerically. The observer's results are compared with and without the history-stack term but both with the observer gain and adaptation gains, projection and initial conditions. For simulation, the system matrices have been selected as:
\begin{equation}
\label{eq:simulation_matrices}
A=
\begin{bmatrix}
0&1\\-1.2&-5
\end{bmatrix},\quad
B=
\begin{bmatrix}
0\\1
\end{bmatrix},\quad
C=
\begin{bmatrix}
1&0
\end{bmatrix},\quad
D=I_2 ,
\end{equation}
and therefore only the first state is measured. The nonlinear functions and true parameter vector are considered as follows:
\begin{align}
\phi(x)
&=-0.25
\begin{bmatrix}
x_1^3&x_2^3
\end{bmatrix}^{\top},\, \theta=
\begin{bmatrix}
0.85&-1.10
\end{bmatrix}^{\top}
\label{eq:simulation_phi}\\
\Phi(x,u)
&=
\begin{bmatrix}
u+0.2\sin(x_1)&\sin(2u)+0.2\sin(x_2)\\
0.005x_1&0.005x_2
\end{bmatrix}
\label{eq:simulation_Phi}
\end{align}
The initial conditions are
$x(0)=[1.20,-0.75]^\top$, $\hat x(0)=[-0.60,0.20]^\top$, and
$\hat\theta(0)=[0.75,-1.00]^\top$, with $\bar\theta=1.5$. 
For the cubic nonlinearity, $e^\top\Delta_\phi\leq0$.
The admissible constants and simulation settings are summarized in
Table~\ref{tab:simulation_parameters}.

\begin{table}[!t]
\centering
\caption{Observer design and simulation parameters.}
\label{tab:simulation_parameters}
\scriptsize
\setlength{\tabcolsep}{2pt}
\renewcommand{\arraystretch}{1.08}
\begin{tabular}{@{}llll@{}}
\toprule
Parameters & Values & Parameters & Values\\
\midrule
Design region
 & $|x_j|,|\hat x_j|\leq5.2$
 & $(\rho,\beta)$
 & $(0,0)$\\
$(\ell_\phi,\alpha)$
 & $(20.28,411.27)$
 & $(\bar\Phi,\ell_\Phi)$
 & $(1.69,0.20)$\\
$P$
 & $0.5I_2$
 & $L$
 & $[100,-0.2]^\top$\\
$(a_e,a_d)$
 & $(1,20)$
 & $(\tau_1,\tau_2,\tau_3)$
 & $(0.5,0.001,0.43)$\\
$Y$
 & $PL$
 & LMI max.\ eigenvalue
 & $-1.0{\times}10^{-3}$\\
$\Psi$
 & $(C^\top)^\dagger P\Phi$
 & $\bar m_\Psi$
 & $1.83{\times}10^{-2}$\\
$\Gamma$
 & $5I_2$
 & $(k_c,\Delta,\sigma_N)$
 & $(2,1~\mathrm{s},0.05)$\\
Candidate start/interval
 & $4~\mathrm{s}/0.5~\mathrm{s}$
 & Residual threshold
 & $5{\times}10^{-3}$\\
$(N,T_w,T_{\mathrm{sim}})$
 & $(5,5,35)~\mathrm{s}$
 & $d(t)$
 & $0$\\
\bottomrule
\multicolumn{4}{@{}l@{}}{$u(t)=e^{-0.075t}\sin(1.25t)$}
\end{tabular}
\end{table}

For this example, ($I-\Pi_C=\diag(0,1)$), and therefore
$\bar m_{\Psi^\star}
=0.0025\sup_{|\hat x_j|\leq5.2}\sqrt{\hat x_1^2+\hat x_2^2}=1.83\times10^{-2}$, confirming the value reported in
Table~\ref{tab:simulation_parameters}. Negative maximum eigenvalue indicates strict LMI feasibility. The absence of disturbances results in finite excitation learning. For the excitation diagnostic, the sliding-window Gramian of the no-stack adaptation regressor is defined for $t\geq T_w$ as:
\begin{equation}
W_\Psi(t)\triangleq
\int_{t-T_w}^{t}
\Psi(\hat x_{\mathrm{ns}}(\tau),u(\tau))^\top
\Psi(\hat x_{\mathrm{ns}}(\tau),u(\tau))d\tau.
\end{equation}

where $\hat x_{\mathrm{ns}}$ is the no-stack state estimate. As shown in Fig.~\ref{fig:excitation_summary}, the Gramian's minimum eigenvalue decreases from $2.7376\times10^{-2}$ to $5.1913\times10^{-4}$, which is indicative of loss of excitation.

It is observed from Fig.~\ref{fig:excitation_summary}that the stack is finite excitation at $T_F$ of 6s with five stored points. This is with $\lambda_{\min}(\hat S_N)$ equal to 0.05 which is more than $\sigma_N$. Moreover, the coupled gain margin is ($a_ek_c\sigma_N-\bar m_\Psi^2$), which itself equals $9.9662\times10^{-2}$.

Offline validation provides the information in the true state scenario where $\lambda_{\min}(S_N^x)=0.0503$ wheras $\rho_G=6.0992\times10^{-4}$ matches with the estimated to true information bound. The obtained values are $\bar e_s=5.2388\times10^{-4}$ and $\|R_N\|=2.8138\times10^{-4}$. The regional invariance condition is satisfied also because its left-hand side is $5.2891$ when compared with $\lambda_{\min}(P)r_e^2$ which equals 8 for when $r_e$ is 4. Additionally, $\max_t\|e(t)\|$ equals 2.03 which becomes smaller than $r_e$ and both the plant and observer trajectories stay in selected design domain. 

Figure~\ref{fig:estimation_errors} indicates that stack activation minimizes parameter error while the no-stack observer loses its learning due to decay of excitation. The right-hand side presents the post-activation state error and terminal parameter error given by:
\begin{equation}
J_{e,F}\triangleq
\sqrt{\frac{1}{T_{\mathrm{sim}}-T_F}
\int_{T_F}^{T_{\mathrm{sim}}}\|e(t)\|^2dt},
\qquad
J_\theta\triangleq
\|\tilde\theta(T_{\mathrm{sim}})\|
\end{equation}

To examine disturbed case covered by Theorem~1, bounded disturbance
$d(t)=[0,\;0.03\sin(1.1(t-T_d))]^\top$ is applied for
$t$ greater than or equal to $T_d$ which is itself 8s, while $d(t)$ is zero before $T_d$. Because $T_F=6$ which is smaller that $T_d$, history stack is frozen prior to the disturbance input. From Fig.~\ref{fig:disturbed_validation}, the error in the estimation of state and parameters is bounded following the disturbance input. %The state errors of the two observers remain comparable, whereas the proposed observer retains a substantially smaller parameter-estimation error than the no-stack ablation. The numerical evaluation also satisfies the joint-error bound in (55) and the regional condition in (48), for which the two sides are $5.06<8$.

\section{Conclusion}

For OSL-QIB nonlinear systems with estimated-state regressors a finite-excitation adaptive observer was designed. By considering the resulting regressor mismatch in the output-integral data, the stability was derived with bounded disturbances and with exponential convergence recovered when disturbance-free and residual-free stored data.

\end{document}